\documentclass[submission,copyright,creativecommons]{eptcs}
\providecommand{\event}{EXPRESS/SOS 2020} % Name of the event you are submitting to
\usepackage[utf8]{inputenc}
\usepackage{newlfont}
\usepackage{tikz-qtree,tikz-qtree-compat}
\usepackage{subcaption}
\usepackage{algorithm, algpseudocode}
\usepackage{graphicx,amsmath,amssymb,wrapfig,hyperref}
\usepackage{booktabs,verbatim}
\usepackage{amsthm}
\usepackage{enumitem}
\usepackage{listings}
\usepackage{xcolor} % for setting colors
\usepackage{semantic}
\usepackage{bussproofs}
\usepackage{graphicx}
\usepackage{stmaryrd}
\usepackage{todonotes}
\usepackage{tikz}
\usetikzlibrary{automata,arrows.meta,positioning}
\usetikzlibrary{matrix}

\theoremstyle{definition}
\newtheorem{theorem}{Theorem}
\newtheorem{definition}[theorem]{Definition}
\newtheorem{lemma}[theorem]{Lemma}
\newtheorem{example}[theorem]{Example}
\newtheorem{corollary}[theorem]{Corollary}

\title{A process algebra with global variables}
\author{
	Mark Bouwman \qquad\quad Bas Luttik \qquad \qquad Wouter Schols \qquad \qquad Tim A.C. Willemse\phantom{tt}
	\institute{Eindhoven University of Technology\\
		Eindhoven, The Netherlands}
	\email{m.s.bouwman@tue.nl \quad s.p.luttik@tue.nl\quad w.r.m.schols@student.tue.nl  \quad t.a.c.willemse@tue.nl}
}
\def\titlerunning{A process algebra with global variables}
\def\authorrunning{M.S. Bouwman, et al.}

\newcommand{\bisim}[1]{~\ensuremath{\underline{\leftrightarrow}_{#1}}~}
\newcommand{\bisimR}[1]{\ensuremath{\mathcal{R}_{#1}}}
\newcommand{\setOp}[1]{\ensuremath{\downarrow(#1)}}
\newcommand{\diamondOp}[1]{\ensuremath{\langle #1 \rangle}}
\newcommand{\boxOp}[1]{\ensuremath{[#1]}}
\newcommand{\ds}[2][]{\ensuremath{\mathalpha{\llbracket{#2}\rrbracket_{#1}}}}

\newcommand{\Var}{\ensuremath{Var}}
\newcommand{\Val}{\ensuremath{\mathcal{V}}}
\newcommand{\Act}{\ensuremath{Act}}
\newcommand{\TL}{\ensuremath{TL}}
\newcommand{\PN}{\ensuremath{PN}}
\newcommand{\multiset}[1]{\ensuremath{\Lbag #1 \Rbag}}

\newcommand{\pref}[1][\lambda]{\ensuremath{\mathalpha{{#1}.}}}
\newcommand{\acmp}{\ensuremath{\mathbin{+}}}
\newcommand{\pcmp}{\ensuremath{\mathbin{\|}}}
\newcommand{\encaps}[1][H]{\ensuremath{\mathalpha{\partial_{#1}}}}

\newcommand{\allow}[1]{\ensuremath{\nabla_{#1}}}
\newcommand{\hide}[1]{\ensuremath{\tau_{#1}}}
\newcommand{\hideF}[1]{\ensuremath{\theta_{#1}}}
\newcommand{\comm}[1]{\ensuremath{\Gamma_{#1}}}
\newcommand{\commF}[1]{\ensuremath{\gamma_{#1}}}
\newcommand{\defeqn}{\ensuremath{\stackrel{\textrm{def}}{=}}}
\newcommand{\step}[1]{\ensuremath{\mathrel{\overset{#1}{\longrightarrow}}}}
\newcommand{\state}[2]{\ensuremath{\langle #1,#2 \rangle}}

\newcommand{\RPref}{\ensuremath{\textsc{(Pref)}}}
\newcommand{\RAsgn}{\ensuremath{\textsc{(Asgn)}}}
\newcommand{\RCon}{\ensuremath{\textsc{(Con)}}}

\newcommand{\RSum}{\ensuremath{\textsc{(Sum)}}}
\newcommand{\RSuml}[1][]{\ensuremath{\textsc{(Sum-l)}^{#1}}}
\newcommand{\RSumr}[1][]{\ensuremath{\textsc{(Sum-r)}^{#1}}}
\newcommand{\RAllow}{\ensuremath{\textsc{(Allow)}}}
\newcommand{\RHide}{\ensuremath{\textsc{(Hide)}}}
\newcommand{\RPar}{\ensuremath{\textsc{(Par)}}}
\newcommand{\RParl}{\ensuremath{\textsc{(Par-l)}}}
\newcommand{\RParr}{\ensuremath{\textsc{(Par-r)}}}
\newcommand{\RComm}{\ensuremath{\textsc{(Comm)}}}
\newcommand{\RRec}[1][]{\ensuremath{\textsc{(Rec)}^{#1}}}
\newcommand{\REnc}{\ensuremath{\textsc{(Enc)}}}

\begin{document}
\maketitle

\begin{abstract}
	In standard process algebra, parallel components do not share a common state and communicate through synchronisation. The advantage of this type of communication is that it facilitates compositional reasoning. For modelling and analysing systems in which parallel components operate on shared memory, however, the communication-through-synchronisation paradigm is sometimes less convenient. In this paper we study a process algebra with a notion of global variable. We also propose an extension of Hennessy-Milner logic with predicates to test and set the values of the global variables, and prove correspondence results between validity of formulas in the extended logic and stateless bisimilarity and between validity of formulas in the extended logic without the set operator and state-based bisimilarity. We shall also present a translation from the process algebra with global variables to a fragment of mCRL2 that preserves the validity of formulas in the extended Hennessy-Milner logic.
\end{abstract}

\section{Introduction}
Communication between parallel components in real world systems takes many forms: packets over a network, inter-process communication, communication via shared memory, communication over a bus, etcetera. Process algebras usually offer an abstract message passing feature. Not all forms of communication fit well in a message passing paradigm, in particular, global variables and other forms of shared memory do not fit in well. In some cases it would be desirable to have global variables as first class citizens. To illustrate this we introduce a small example.

\begin{example}\label{exm:car-intro}
	Consider a traffic light and a car approaching a junction. If the light is green the car performs an action $drive$ and moves past the junction. If the light is red the car performs an action $brake$ and stops. Once the traffic light is green again the car performs the action $drive$. The specification should result in the following LTS.
	
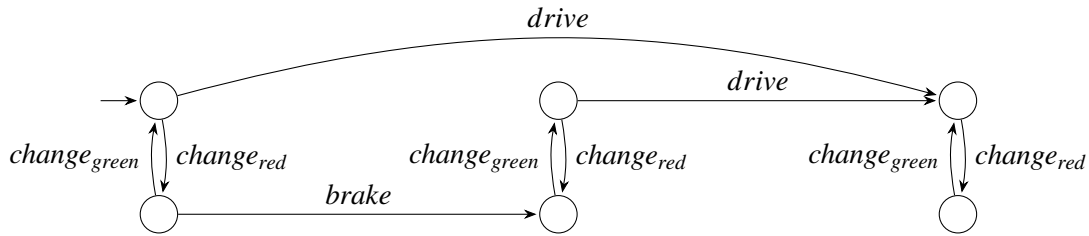
\begin{figure}[H]
		\centering
		\begin{tikzpicture}[
		>=Stealth,
		shorten >=1pt,
		auto,
		node distance=4.8 cm,
		scale = 1,
		transform shape,
		state/.style={inner sep=5pt, draw=black, circle}
		]
		
		\node[initial,state,initial text=]  (start)                    {};
		\node[state]                        (redbefore)    [below=1 cm of start]   {};
		\node[state]                        (redstopped)    [right=of redbefore]   {};
		\node[state]                        (greenstopped)    [right=of start]   {};
		\node[state]                        (finishedgreen)    [right=of greenstopped]   {};
		\node[state]                        (finishedred)    [right=of redstopped]   {};
		
		\path[->] (start) edge [bend left=15] node    {$drive$}           (finishedgreen)
		(start) edge [bend left=10]  node    {$change_{red}$}           (redbefore)
		(redbefore) edge node    {$brake$}           (redstopped)
		(redbefore) edge [bend left=10] node  {$change_{green}$}           (start)
		(redstopped) edge [bend left=10] node {$change_{green}$}           (greenstopped)
		(greenstopped) edge [bend left=10] node {$change_{red}$}           (redstopped)
		(finishedgreen) edge [bend left=10] node    {$change_{red}$}           (finishedred)
		(finishedred) edge [bend left=10] node    {$change_{green}$}           (finishedgreen)
		(greenstopped) edge node {$drive$}           (finishedgreen);
		
		\end{tikzpicture}
	\end{figure}
	It would be natural to model the car and the traffic light as two parallel components. The car and the traffic light need to communicate so that the car can make a decision to drive or brake depending on the current state of the traffic light. Typically, such global information is modelled by introducing an extra parallel component that maintains the global information, in this case the colour of the traffic light. However, taking that approach we obtain a different LTS, which has an extra transition modelling the communication of car and traffic light with the extra component. Moreover, one must take care that decisions to drive or brake are made on the basis of up-to-date
information, e.g., by implementing a protocol that locks the additional component. In many cases it is realistic that the observed value is no longer up to date and in some cases we are also interested in analysing the consequences of this. In other cases however, we might want to abstract from such complications. When the information is constantly available to the observers, as is the case with a traffic light, we have even more reason to not introduce separate transitions communicating the global information.
	
	In some process algebras it is possible to define a communication function that specifies that a $drive$ action is the result of a communication between two actions of parallel components, e.g. a $drive\_if\_green$ action from the car and a $signal\_green$ action from the traffic light. This is somewhat unnatural as the traffic light does not really actively participate in the driving of the car. Moreover, if we introduce a second car that only wants to drive when the traffic light is red we would need to change the communication function, even though the communication of information does not essentially change. It would be better to let the colour of the traffic light be in a global variable. In that way the behaviour of the traffic light and cars acting upon information from the traffic light is more separated, we obtain a better separation of concerns.
\end{example}

In the early days of process algebra, doing away with global variables in favour of message passing and local variables was an important step to further develop the field \cite{DBLP:journals/tcs/Baeten05}. Since then there have, nevertheless, been some efforts to reintroduce notions of globally available data. 

In \cite{proces_alg_book} propositional signals and a state operator are presented. The state operator tracks the value of some information. Based on the current value a number of propositions can be signalled to the rest of the system. In the example of the traffic lights the process modelling the traffic light could track the state of the light, emitting signals such as $lightGreen$ and $\neg lightRed$, which can in turn be used as conditions in the process modelling the car. In this approach the value of the global variable is not communicated directly, which restricts conditions based on global variables to propositional logic. 

Other approaches, such as the one presented in \cite[Chapter 19]{DBLP:series/txcs/Roscoe10}, model global variables as separate parallel processes and use a protocol to ensure only one process accesses a global variable at the time. This approach introduces extra internal steps, which increase the statespace. Moreover, it introduces divergence when a process locks a global variable, reads the value, concludes that it cannot make a step and unlocks the variable again. 

Formalisms based on Concurrent Constraint Programming (CCP) \cite{DBLP:conf/popl/SaraswatRP91} have global data at their core. In CCP a central store houses a set of constraints. Concurrent processes may \emph{tell} a constraint, adding it to the global store or \emph{ask} a constraint, checking whether it is entailed by the constraints in the store. An ask will block until other processes have added sufficient constraints to the store. Process calculi based on the coordination language LINDA \cite{DBLP:conf/coordination/NicolaP96} also use global data. In these process calculi there is a global set of data elements. Similarly to CCP, processes may tell a data element (adding it to the global set) or ask a data element (checking whether it is in the set). Additionally, processes may \emph{get} an element, removing it from the data set. LINDA does not have a concept of variables, just a central set of data elements. Generally, process calculi based on CCP or LINDA do not allow asking a constraint/data element and acting upon the information in a single step. 

The goal of our work is to propose and study (i) a process calculus with global variables (ii) a modal logic that can refer to the values of global variables (iii) an encoding in an existing process calculus and logic with tool support. In this paper we propose a simple process calculus where every component of the system can access the current value of global variables directly. We define appropriate notions of equivalence for our process calculus. Our first contribution is an extension of the Hennessy-Milner Logic (HML) with two new operators that is strong enough to differentiate non-equivalent process expressions. Our second contribution is an encoding of our process algebra in mCRL2 and our extended logic in standard HML. This encoding is such that the translated formula holds for the translated process expression if and only if the original formula holds for the original process expression.

This paper is organised as follows. In Section \ref{sec:GlobalVarDef} we define a simple process algebra with global variables. In Section \ref{sec:equivalence-process} we give appropriate notions of equivalence for our process algebra. We continue by defining an extension of the Hennessy-Milner Logic in Section \ref{sec:HMLogic} and relating it to our equivalence notions in Section \ref{sec:HML-correspondence}. In Section \ref{sec:translation} we show how our process algebra with global variables can be encoded in mCRL2. Sections \ref{sec:discussion} and \ref{sec:conclusion} discuss the results and conclude this work.

\section{A simple process algebra with global variables}
\label{sec:GlobalVarDef}
In this section we will introduce a process algebra with global variables and its semantics. For convenience we will, in this paper, assume a single data domain $D$. We will use \Var{} to denote the finite set of global variable names.

\newcommand{\PexprG}{\ensuremath{\mathcal{P}}}

We presuppose a set of actions $\Act{}$ and derive a set of transition labels $\TL{} \defeqn \Act \cup \{assign(v,d) \mid v\in \Var{} \land d\in D\}$. We also presuppose a set of process names \PN. The set of process expressions \PexprG{} is generated by the following grammar containing \emph{action\ prefix, inaction, choice, parallelism, encapsulation, recursion and conditionals}:

$$P:=\lambda.P\ |\ \delta\ |\ P+P\ |\ P \| P\ |\ \partial_B(P)\ |\ X\ |\ (v = d) \xrightarrow{} P,$$

\noindent where $\lambda \in \TL{}$, $B \subseteq \Act$, $X \in \PN{}$, $v \in \Var{}$ and $d\in D$. Inaction is similar to the process constant 0 in, for example, CCS \cite{DBLP:books/daglib/0067019} and TCP \cite{proces_alg_book}. Our process algebra supports recursion because we also define a recursive specification $E$ defining the process names. Let a recursive equation be an equation of the form $X \defeqn t$ with $X \in \PN{}$ and $t$ a process expression in \PexprG. A recursive specification contains one recursive equation $X \defeqn t$ for every $X \in \PN{}$. Every recursive specification should be guarded. This means that every occurrence of $X$ in $t$ is in the scope of an action prefix. For communication between parallel processes we use an ACP style communication function. We presuppose a binary communication function on the set of actions, i.e., a partial function $\gamma:\Act \times \Act \rightharpoonup \Act$ that is commutative and associative. We only allow handshakes (communication between two parties): if $\gamma(a,b) = c$ then $\gamma(c,d)$ is undefined for every $d$.

Let $\mathcal{V}$ be the set of all functions $\Var{} \rightarrow D$, i.e. the set of all valuations. Let $V \in \mathcal{V}$; we denote by $V[v \mapsto d]$ the assignment defined, for all $v' \in \Var{}$ by: 
$$
	V[v \mapsto d](v') =
	\left\{
	\begin{array}{ll}
	d  & \mbox{if } v' = v \\
	V(v') & \mbox{if } v' \neq v
	\end{array}
	\right.
$$

\noindent In Definition \ref{def:TransitionSystem} we give the usual definition of Labelled Transition Systems (LTSs).

\begin{definition}
	A Labelled Transition System (LTS) is a tuple $(S,L,\xrightarrow{},s)$, where
	\begin{itemize}
		\item $S$ is a set of states,
		\item $L$ is a set of transition labels,
		\item ${\rightarrow} \subseteq S\times L \times S$ is the transition relation,
		\item $s \in S$ is the initial state.
	\end{itemize}
	\label{def:TransitionSystem}
\end{definition}
We now want to associate an LTS with the process algebra. As the behaviour of a process expression depends on the valuation of global variables a state is a pair $\langle P, V\rangle$ of a process expression $P$ and a valuation function $V$. The set of states is $\PexprG \times \mathcal{V}$. The transition relation is the least relation on states satisfying the rules of the structural operational semantics (see Table \ref{tab:os}). 

\begin{table}\small\center
	\begin{tabular}{c} \hline \\
		$\RPref\
		\inference{}{\state{\pref[a]P}{V}\step{a}\state{P}{V}}$
		\qquad\quad $\RAsgn\
		\inference{}{\state{assign(v,d).P}{V}\step{assign(v,d)}\state{P}{V[v \mapsto d]}}$
		\\ \\
		$\RCon\ 
		\inference{\state{P}{V}\step{\lambda}\state{P'}{V'}}{\state{(v=d)\rightarrow P}{V} \step{\lambda}\state{P'}{V'}} V(v) = d$
		\qquad\quad $\RRec\ \inference{\state{P}{V}\step{\lambda}\state{P'}{V'}}{\state{X}{V}\step{\lambda}\state{P'}{V'}}X\defeqn P$
		\\ \\
		$\RSuml\
		\inference{\state{P}{V}\step{\lambda}\state{P'}{V'}}%
		{\state{{P\acmp Q}}{V}\step{\lambda}\state{P'}{V'}}$\qquad\quad
		$\RSumr\
		\inference{\state{Q}{V}\step{\lambda}\state{Q'}{V'}}%
		{\state{{P\acmp Q}}{V}\step{\lambda}\state{Q'}{V'}}$
		\\ \\
		$\RParl\ \inference{\state{P}{V}\step{\lambda}\state{P'}{V'}}{\state{P\pcmp
				Q}{V}\step{\lambda}\state{P'\pcmp Q}{V'}}$\qquad\quad
		$\RParr\ \inference{\state{Q}{V}\step{\lambda}\state{Q'}{V'}}{\state{P\pcmp Q}{V}\step{\lambda}\state{P\pcmp Q'}{V'}}$
		\\ \\
		$\RComm\ \inference{\state{P}{V}\step{a}\state{P'}{V} & \state{Q}{V}\step{b}\state{Q'}{V}}{\state{P\pcmp Q}{V}\step{c}\state{P'\pcmp Q'}{V}}\gamma(a,b)=c$
		\\ \\
		$\REnc\ \inference{\state{P}{V}\step{\lambda}\state{P'}{V'}}{\state{\encaps[B](P)}{V}\step{\lambda}\state{\encaps[B](P')}{V'}}\lambda\notin
					B$ 
		\\ \\
		\hline
	\end{tabular}
	\caption{Structural operational semantics.}\label{tab:os}
\end{table}

Note that we only allow processes to synchronise on actions and not on assignments. This design decision was made since assignments change the valuation function, whereas actions cannot change the valuation. When two processes synchronise on assignments then it is not clear what the resulting effect on the value of the variable should be.

\begin{example}
	Consider the interaction between a car and a traffic light controller (TLC). The TLC sets the colour of a traffic light which the driver of the car acts upon. There is one global variable $t$ and the data domain consists of two elements $D = \{green, red\}$. The recursive specification consists of two process equations, given below.
	
	$$CAR \defeqn ((t = green) \xrightarrow{} drive.\delta)\ +\ ((t = red) \xrightarrow{} brake.((t = green) \rightarrow drive.\delta)) $$
	$$TLC \defeqn ((t = green) \xrightarrow{} assign(t,red).TLC)$$ $$+ ((t = red) \xrightarrow{} assign(t,green).TLC)$$
	
	\noindent Using the SOS we can derive an LTS with $\langle CAR || TLC, V \rangle$ as initial state, where $V(t) = green$. Note that this LTS is isomorphic to the LTS presented in Example \ref{exm:car-intro}. We only show the states reachable from the initial state. The initial state is marked with an arrow pointing to it.
	
	
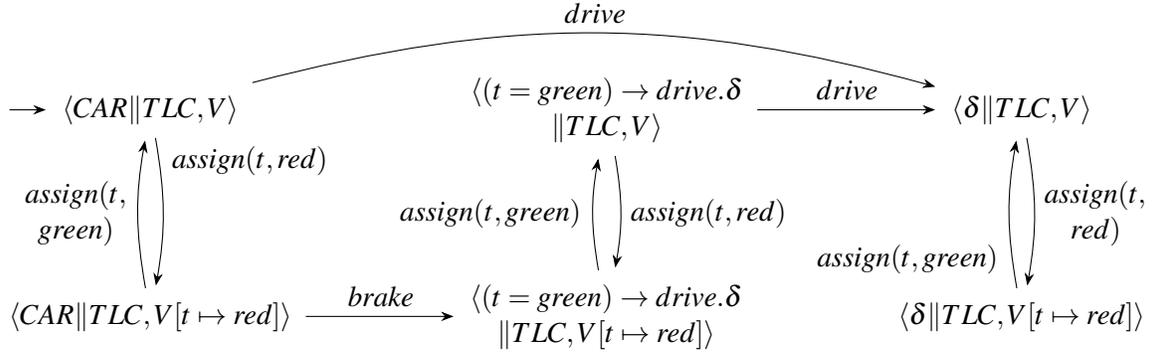
\begin{figure}[H]
			\centering
			\begin{tikzpicture}[
			>=Stealth,
			shorten >=1pt,
			auto,
			node distance=1.7 cm,
			scale = 1,
			transform shape,
			state/.style={rectangle, inner sep=5pt}
			]
			
			\node[initial,state,initial text=]  (start)                    {$\langle CAR \|TLC, V \rangle$};
			\node[state, align=center]                        (redbefore)    [below=2 cm of start]   {$\langle CAR \|TLC, V[t \mapsto red] \rangle$};
			\node[state, align=center]                        (redstopped)    [right= 2cm of redbefore]   {$\langle (t = green) \rightarrow drive.\delta $\\$\|TLC, V[t \mapsto red] \rangle$};
			\node[state, align=center]                        (greenstopped)    [left= 2.3cm of finishedgreen]   {$\langle (t = green) \rightarrow drive.\delta $\\$\|TLC, V \rangle$};
			
			\node[state]                        (finishedred)    [right=of redstopped]   {$\langle \delta \|TLC, V[t \mapsto red] \rangle$};
			\node[state]                        (finishedgreen)    [above=2cm of finishedred]   {$\langle \delta \|TLC, V\rangle$};
			
			\path[->] (start) edge [bend left=15] node    {$drive$}           (finishedgreen)
			(start) edge [bend left=10, pos=0.3]  node    {$assign(t,red)$}           (redbefore)
			(redbefore) edge node    {$brake$}           (redstopped)
			(redbefore) edge [bend left=10, align=center] node  {$assign(t, $\\$ green)$}           (start)
			(redstopped) edge [bend left=10] node {$assign(t, green)$}           (greenstopped)
			(greenstopped) edge [bend left=10] node {$assign(t, red)$}           (redstopped)
			(finishedgreen) edge [bend left=10, align=center] node    {$assign(t, $\\$ red)$}           (finishedred)
			(finishedred) edge [bend left=10, pos=0.35] node    {$assign(t, green)$}           (finishedgreen)
			(greenstopped) edge node {$drive$}           (finishedgreen);
			
			\end{tikzpicture}
		\end{figure}
\end{example}

\section{Equivalence of process expressions}\label{sec:equivalence-process}
We will examine equivalence relations in the context of global variables. To start we note that we can examine equivalence on two levels: on the level of process expressions and on the level of pairs of process expression together with an initial valuation (from which we can derive an LTS). We begin by exploring the equivalence of process expressions.

We require of the equivalence relation on process expressions that if $P$ and $Q$ are equivalent then we can safely replace $P$ with $Q$ in any larger process expression. In other words, the equivalence relation on process expressions should be a congruence for the process algebra. 

Typically, equivalence of process expressions is established by a notion of bisimilarity. Most variants of bisimulation only consider the labels on transitions. Strong bisimilarity (defined in Definition \ref{def:strong-bisimilarity}) is, however, not a congruence for our process algebra, which we will demonstrate with an example. 

\begin{definition}{Strong bisimilarity:}
	\label{def:strong-bisimilarity}
	A relation $\bisimR{} \subseteq S \times S$, where $S$ is the set of states of an LTS, is a strong bisimulation relation if and only if for all states $s$ and $t$ and labels $\lambda$ we have $(s,t) \in \bisimR{}$ implies that
	\begin{itemize}
		\item for all states $s'$: $s \xrightarrow{\lambda} s'$ implies there exists a state $t'$ such that $t \xrightarrow{\lambda} t'$ and $(s', t') \in \bisimR{}$,
		\item for all states $t'$: $t \xrightarrow{\lambda} t'$ implies there exists a state $s'$ such that $s \xrightarrow{\lambda} s'$ and $(s', t') \in \bisimR{}$.
	\end{itemize}
	Two states $s$ and $t$ are strongly bisimilar, denoted by $s \bisim{} t$, if and only if there exists a strong bisimulation relation $\bisimR{}$ such that $(s, t) \in \bisimR{}$.
\end{definition}

\begin{example}\label{exmp:congruence}
Consider process expressions $P = (v=0)\rightarrow a.\delta$ and $Q = a.\delta$. Note that $P$ and $Q$ are simply abbreviations of process expressions, not process names. Let $V$ map a global variable $v$ to $0$ and $D = \{0,1\}$. The reachable fragments of the LTSs with $\langle P,V\rangle$ and $\langle Q,V\rangle$ as initial state are shown in Figure \ref{fig:exampleHM1}. 

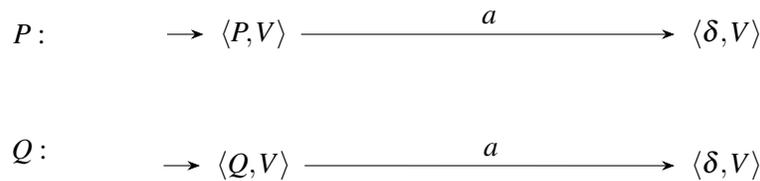
\begin{figure}[H]
	\centering
	\begin{tikzpicture}[
	>=Stealth,
	shorten >=1pt,
	auto,
	node distance=1 cm,
	scale = 1,
	transform shape,
	state/.style={rectangle, inner sep=5pt}
	]	
	\node[state, align=center,initial,initial text=]  (Pcond0)  {$\langle P,V \rangle$};
	%\node[state, align=center]  (Pcond1)  [below=0.5cm of Pcond0]    	  {$\langle (v=0)\rightarrow a.\delta,$\\$val[v \mapsto 1] \rangle$};
	\node[state,initial,initial text=]  (Qa0) [below=1cm of Pcond0]         {$\langle Q,V \rangle$};
	%\node[state]  	(Qa1)       		[below=4cm of Pcond0]         {$\langle a.\delta,val[v \mapsto 1] \rangle$};
	
	\node[state]  (Pdelta0)       [right=5cm of Pcond0]          {$\langle \delta,V \rangle$};
	%\node[state]  (Pdelta1)		  [below right=1cm and 5cm of Pcond0]    	  	 {$\langle \delta,val[v \mapsto 1] \rangle$};
	\node[state]  (Qdelta0)		 [below=1cm of Pdelta0]        		  {$\langle \delta,V \rangle$};
	%\node[state]  (Qdelta1)      [below right=4cm and 5cm of Pcond0]         	  {$\langle \delta,val[v \mapsto 1] \rangle$};
	
	\node (P)  [left=2cm of Pcond0] {$P:$};
	\node (Q)  [below=1cm of P] {$Q:$};

	\path[->]
	(Pcond0) edge   node[pos=0.5]     {$a$}           (Pdelta0)
	(Qa0) edge   node[pos=0.5]     {$a$}           (Qdelta0);
	%(Qa1) edge   node[pos=0.5]     {$a$}           (Qdelta1);
	\end{tikzpicture}
	\caption{Part of the transition system space of $P$ and $Q$}
	\label{fig:exampleHM1}
\end{figure}

Processes $P$ and $Q$ seem behaviourally equivalent looking at the reachable transitions, the states $\langle P,V\rangle$ and $\langle Q,V \rangle$ are in fact strongly bisimilar. The problem arises when we add a parallel component that can assign a different value to the global variable. Let us consider the process expression $R = assign(v, 1).\delta$. The reachable fragments of the LTSs with $P\|R$ and $Q\|R$ as initial state are shown in Figure \ref{fig:exampleHM2}. Clearly $P\|R$ and $Q\|R$ are not strongly bisimilar and therefore strong bisimilarity is not a congruence for our process algebra. 

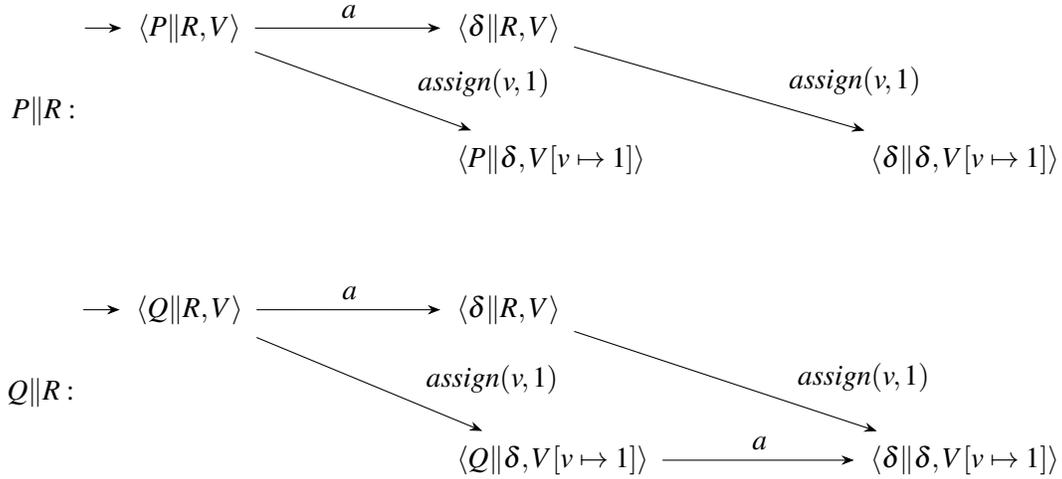
\begin{figure}[H]
	\centering
	\begin{tikzpicture}[
	>=Stealth,
	shorten >=1pt,
	auto,
	node distance=1 cm,
	scale = 1,
	transform shape,
	state/.style={rectangle, inner sep=5pt}
	]	
	\node[state, align=center,initial,initial text=]  (PPR0)  {$\langle P\|R,V \rangle$};
	%\node[state, align=center]  (PPR1)  [below=1cm of PPR0]    	  {$\langle P\|R,V[v \mapsto 1] \rangle$};
	\node[state]  (PR0)       [right=2.5cm of PPR0]          {$\langle \delta\|R,V \rangle$};
	\node[state]  (PP1)		  [below right=1cm and 2.5cm of PPR0]    	  	 {$\langle P\|\delta,V[v \mapsto 1] \rangle$};
	\node[state]  (Pdelta1)		  [below right=1cm and 8cm of PPR0]    	  	 {$\langle \delta\|\delta,V[v \mapsto 1] \rangle$};
	
	\node[state,initial,initial text=]  (QQR0) [below=3cm of PPR0]         {$\langle Q\|R,V \rangle$};
	%\node[state]  (QQR1)       		[below=5cm of PPR0]         {$\langle Q\|R,V[v \mapsto 1] \rangle$};
	\node[state]  (QR0) [below right=3cm and 2.5cm of PPR0]         {$\langle \delta\|R,V \rangle$};
	\node[state]  (QQ1)       		[below right=5cm and 2.5cm of PPR0]         {$\langle Q\|\delta,V[v \mapsto 1] \rangle$};
	\node[state]  (Qdelta1)      [below right=5cm and 8cm of PPR0]         	  {$\langle \delta\|\delta,V[v \mapsto 1] \rangle$};
	
	\node (P)  [below left=0.4cm and 0.5cm of PPR0] {$P\|R:$};
	\node (Q)  [below left=0.4cm and 0.5cm of QQR0] {$Q\|R:$};
	
	\path[->]
	(PPR0) edge   node[pos=0.5]     {$a$}           (PR0)
	%(PPR1) edge   node[pos=0.5,below]     {$assign(v:=1)$}           (PP1)
	(PPR0) edge   node[pos=0.7]     {$assign(v, 1)$}           (PP1)
	(PR0) edge   node[pos=0.7]     {$assign(v, 1)$}           (Pdelta1);
	
	\path[->]
	(QQR0) edge   node[pos=0.5]     {$a$}           (QR0)
	%(QQR1) edge   node[pos=0.5,below]     {$assign(v:=1)$}           (QQ1)
	(QQR0) edge   node[pos=0.7]     {$assign(v, 1)$}           (QQ1)
	(QR0) edge   node[pos=0.7]     {$assign(v, 1)$}           (Qdelta1)
	(QQ1) edge   node[pos=0.5]     {$a$}           (Qdelta1);
	\end{tikzpicture}
	\caption{Part of the transition system space of $P$ and $Q$}
	\label{fig:exampleHM2}
\end{figure}
\end{example}

We will use the notion of stateless-bisimilarity, defined in \cite{DBLP:journals/iandc/MousaviRG05}, as an equivalence relation on process expressions. In essence, stateless-bisimilarity relates process expressions that behave the same under any valuation. 

\begin{definition}{Stateless-bisimilarity:}
	\label{def:stateless-bisimilarity}
	A relation $\bisimR{sl} \subseteq \PexprG \times \PexprG$ is a stateless bisimulation relation if and only if for all process expressions $P$ and $Q$ and labels $\lambda$ we have $(P,Q) \in \bisimR{sl}$ implies that
	\begin{itemize}
		\item for all process expressions $P'$ and valuation functions $V,V'\in\mathcal{V}$: $\langle P,V \rangle \xrightarrow{\lambda} \langle P',V' \rangle$ implies there exists a process expression $Q'$ such that $\langle Q,V \rangle \xrightarrow{\lambda} \langle Q',V' \rangle$ and $(P',Q') \in \bisimR{sl}$,
		\item for all process expressions $Q'$ and valuation functions $V,V'\in\mathcal{V}$: $\langle Q,V \rangle \xrightarrow{\lambda} \langle Q',V' \rangle$ implies there exists a process expression $P'$ such that $\langle P,V \rangle \xrightarrow{\lambda} \langle P',V' \rangle$ and $(P',Q') \in \bisimR{sl}$.
	\end{itemize}
	Two process expressions $P$ and $Q$ are stateless bisimilar, denoted by $P \bisim{sl} Q$, if and only if there exists a stateless bisimulation relation $\bisimR{sl}$ such that $(P,Q) \in \bisimR{sl}$.
\end{definition}

The deduction system of our process algebra is in process-tyft format from which it follows that stateless-bisimilarity is a congruence \cite{DBLP:journals/iandc/MousaviRG05}. 

In the case that global variables cannot be changed by the environment and we have a specific initial valuation in mind we might not care about the behaviour under valuations that will never occur. To that end we use state-based bisimilarity \cite{DBLP:journals/iandc/MousaviRG05}, which is defined on states rather than process expressions.  

\begin{definition}{State-based bisimilarity:}
	\label{def:strong-data-bisimilarity}
	A relation $\bisimR{sb} \subseteq (\PexprG \times \Val{}) \times (\PexprG \times \Val{})$ is a state-based bisimulation relation if and only if for all states $\langle P, V_1\rangle$ and $\langle Q, V_2 \rangle$ and labels $\lambda$ we have $(\langle P, V_1 \rangle, \langle Q, V_2 \rangle ) \in \bisimR{sb}$ implies that $V_1 = V_2$ and
	\begin{itemize}
		\item for all process expressions $P'$ and valuation functions $V'$: $\langle P,V_1 \rangle \xrightarrow{\lambda} \langle P',V' \rangle$ implies there exists a process expression $Q'$ such that $\langle Q,V_2 \rangle \xrightarrow{\lambda} \langle Q',V' \rangle$ and $(\langle P', V' \rangle, \langle Q', V' \rangle ) \in \bisimR{sb}$,
		\item for all process expressions $Q'$ and valuation functions $V'$: $\langle Q,V_2 \rangle \xrightarrow{\lambda} \langle Q',V' \rangle$ implies there exists a process expression $P'$ such that $\langle P,V_1 \rangle \xrightarrow{\lambda} \langle P',V' \rangle$ and $(\langle P', V' \rangle, \langle Q', V' \rangle ) \in \bisimR{sb}$.
	\end{itemize}
	Two states $\langle P, V_1\rangle$ and $\langle Q, V_2 \rangle$ are state-based bisimilar, denoted by $\langle P, V_1\rangle \bisim{sb} \langle Q, V_2 \rangle$, if and only if there exists a state-based bisimulation relation $\bisimR{sb}$ such that $(\langle P, V_1 \rangle, \langle Q, V_2 \rangle ) \in \bisimR{sb}$.
\end{definition}

State-based bisimilarity is not a congruence for our process algebra, the problem shown in Example \ref{exmp:congruence} applies.

There is a relation between stateless-bisimilarity and state-based bisimilarity. For any two process expressions $P$ and $Q$ we have that if $P \bisim{sl} Q$ then also $\langle P,V \rangle \bisim{sb} \langle Q,V \rangle$ for all valuations $V \in \Val$ \cite{DBLP:journals/iandc/MousaviRG05}.

State-based bisimilarity distinguishes LTSs on the valuation that is in a state: two states that are strongly bisimilar may not be state-based bisimilar due to differences in valuations in reachable states. This takes into account that the value of global variables may be essential to the modelled system and may be visible to the environment. 

\section{Hennessy-Milner logic}
\label{sec:HMLogic}
In order to reason about properties of a process expression or system specification we define a logic. Standard Hennessy-Milner Logic (HML) \cite{HennessyM85} is insufficient for our purpose, for two reasons. The first reason is that we would like to conveniently refer to global variables in the logic. The second reason for extending the logic is that we want a correspondence between the logic and stateless bisimilarity. Process expressions $a.\delta$ and $ (v = 0) \xrightarrow a.\delta$ are not stateless bisimilar but in the case that we have a valuation function $V$ that maps $v$ to 0 then states $\langle a.\delta, V \rangle$ and $\langle (v = 0) \xrightarrow a.\delta, V \rangle$ cannot be distinguished using HML.

We extend HML with two new operators. The first operator is the check operator $(v=e)$. This operator returns a boolean which is true if and only if a global variable $v$ has value $e$. The second operator is the set operator $\setOp{v := e}$. The set operator sets the value of a global variable $v$ to $e$. This results in the following syntax for our logic:

$$\phi := true\ |\ false\ |\ (v = e) \ |\ \neg \phi\ |\ \phi \wedge \phi\ |\ \phi \vee \phi\ |\ 
\diamondOp{T}\phi\ |\ \boxOp{T}\phi\ |\ \setOp{v:=e}\phi$$

where $T$ is a nonempty finite set of transition labels. Depending on whether we include the check operator, the set operator or both, we will refer to the logic with HML$^{check}$, HML$^{set}$ or HML$^{check+set}$, respectively.

The formula $true$ holds in every state and $false$ holds in no states. The operators $\neg,\wedge,\vee$ have their usual meaning. The diamond operator $\diamondOp{T}\phi$ is true in a state $s$ if and only if a transition $s \xrightarrow{\lambda} s'$ exists where $\phi$ holds in $s'$ and $\lambda \in T$. The box operator $[T]\phi$ holds in a state $s$ if and only if for every state $s'$ and transition label $\lambda \in T$ we have $s \xrightarrow{\lambda} s'$ implies $\phi$ holds in $s'$. 

The check operator $(v=e)$ is true in a state $\langle P,V\rangle$ if and only if $V(v)=e$. The set operator $\downarrow$ indicates that $\setOp{v:=e}\phi$ is true in all states $\langle P,V\rangle$ if and only if $\phi$ is true in $\langle P,V[v \mapsto e]\rangle$. Note that the set operator allows us to reason about parts of an LTS that are not reachable from the initial state. Further note that the set operator allows us to distinguish $\langle a.\delta, V \rangle$ and $\langle (v = 0) \xrightarrow a.\delta, V \rangle$ even if $V(v) = 0$ the formula $\setOp{v:=1}\langle \{a\} \rangle true$ distinguishes them. We will use the notation $\setOp{V}$,  $V \in \mathcal{V}$, to set the value of all global variables to the value specified by $V$. This is a shorthand for a sequence of regular set operations. Note that the number of global variables is finite and the order of set operations is irrelevant in the sequence of set operations as each sets a different variable.

\subsection{Semantics}
In this section we will define semantic rules to obtain all states that satisfy a HML$^{check+set}$ formula. We have obtained the semantics of standard HML from \cite{sv_book}. Let $\phi$ be a modal formula, let $(S,L,\xrightarrow{},s)$ be an LTS. We inductively define the interpretation of $\phi$, notation $\ds{\phi}$, where $\ds{\phi}$ contains all states $u\in S$ where $\phi$ is true. Note that the check and set operators are only defined for LTSs where states consist of both a process expression and a valuation.
\begin{center}
	\begin{tabular}{lll}
		$\ds{true} $&=&$ S$\\
		$\ds{false} $&=&$ \emptyset$\\
		$\ds{v = e} $&=&$ \{\langle P,V\rangle\ \in S\ |\ V(v) = e\}$\\
		$\ds{\neg \phi} $&=&$ S \setminus \ds{\phi} $\\
		$\ds{\phi \wedge \phi'} $&=&$ \ds{\phi} \cap\ds{\phi'} $\\
		$\ds{\phi \vee \phi'} $&=&$ \ds{\phi} \cup \ds{\phi'} $\\
		$\ds{\langle T \rangle \phi} $&=&$ \{u \in S\ |\ \exists u' \in \ds{\phi}, \lambda \in T: u\xrightarrow{\lambda}u'  \} $\\
		$\ds{[ T ] \phi} $&=&$ \{u \in S\ |\ \forall u' \in S, \lambda \in T:  (u\xrightarrow{\lambda}u') \Rightarrow u' \in \ds{\phi}  \} $\\
		$\ds{\setOp{v := e}\phi} $&=&$ \{\langle P, V\rangle \in S\ |\ \langle P,V[v\mapsto e]\rangle \in \ds{\phi}  \} $\\
	\end{tabular}    
\end{center}

\section{Relation logic and bisimilarity}\label{sec:HML-correspondence}
There is a nice correspondence between strong bisimilarity and HML: two states in an LTS are strongly bisimilar if and only if they satisfy the same HML formulas \cite{HennessyM85}. This correspondence is often called the Hennessy-Milner theorem. We would like a similar correspondence between process expressions and states and the extended HML. First, we introduce the notion of an image-finite process. As an LTS contains all possible process expressions (and valuations) we want to impose image-finiteness only for reachable states and process expressions, so we start by defining reachability.

\begin{definition}{Reachability states:}
	A state $s'$ is reachable from a state $s$ if there exist states $s_0, \dots,s_n$ and labels $\lambda_1,\dots,\lambda_n$ such that $s=s_0 \land s_0 \step{\lambda_1}s_1 \land\cdots\land s_{n-1} \step{\lambda_n}s_n\land  s_n = s'$.
\end{definition}

\begin{definition}{Reachability process expressions:}
	Process expression $P’$ is reachable from a process expression $P$ if there exist processes  $P_0, \dots,P_n$ and labels $\lambda_1,\dots,\lambda_n$ such that $P=P_0 \land \exists_{V_0,V_1} \state{P_0}{V_0} \step{\lambda_1}\state{P_1}{V_1} \land\cdots\land\exists_{V_{n-1},V_{n}} \state{P_{n-1}}{V_{n-1}} \step{\lambda_n}\state{P_n}{V_{n}}\land  P_n = P’$.
\end{definition}

\begin{definition}{Image-finiteness:}
	A state $\langle P, V \rangle$ is image finite if and only if the set $\{\langle P', V' \rangle | \langle P, V \rangle \xrightarrow{\lambda} \langle P', V' \rangle \}$ is finite for every label $\lambda$. A process expression $P$  is image finite if and only if for every process expression $P'$ reachable from $P$ and every valuation $V$ the state $\langle P', V \rangle$ is image-finite.
\end{definition}

\noindent An example of a state that is not image finite is $\langle A, V \rangle$, with $A \defeqn a.\delta || A$.

We can now prove the following two correspondences on the level of process expressions and states.
\begin{theorem}\label{thm:relation-logic-stateless}
Let $P$ and $Q$ be two image-finite process expressions. Then $P \bisim{sl} Q$ if and only if for all valuations $V$  and all HML$^{check+set}$ formulas $\phi$ we have that $\langle P, V \rangle \in \ds{\phi} \Leftrightarrow \langle Q,V \rangle \in \ds{\phi}$.
\end{theorem}
\begin{proof}
	We prove the two implications separately. To prove the implication from left to right assume $P \bisim{sl} Q$. The proof that for some HML$^{check+set}$ formula $\phi$ we have that $\langle P, V \rangle \in \ds{\phi}$ if and only if $\langle Q, V \rangle \in \ds{\phi}$ is straightforward by induction on the structure of $\phi$. 
	
	For the implication from right to left we assume that $\langle P,V \rangle$ and $\langle Q,V \rangle$ satisfy exactly the same formulae in HML$^{check+set}$. We shall prove that $\langle P, V \rangle \bisim{sl} \langle Q,V \rangle$. To this end, note that it is sufficient to show that the relation 
	$$\bisimR{sl} = \{(T,U) | T,U \in \PexprG{} \text{ and } \forall_{V \in \mathcal{V}} \langle T, V \rangle \text{ and }  \langle U, V \rangle \text{ satisfy the same HML$^{check+set}$ formulae}\}$$
	is a stateless bisimulation relation. Assume that $T \bisimR{sl} U$ and $\langle T, \overline{V} \rangle \xrightarrow{\lambda} \langle T', V' \rangle$ for some valuation $\overline{V}$. We shall now argue that there is a process $U'$ such that $\langle U, \overline{V} \rangle \xrightarrow{\lambda} \langle U', V' \rangle$ and $T' \bisimR{sl} U'$. Since $\bisimR{sl}$ is symmetric, this suffices to establish that $\bisimR{sl}$ is a stateless bisimulation relation.
	
	Now assume, towards a contradiction, that there is no $\langle U', V' \rangle$ such that $\langle U, \overline{V} \rangle \xrightarrow{\lambda} \langle U', V' \rangle$ and for all valuations $V\in \mathcal{V}$, $\langle U', V \rangle$ satisfies the same HML$^{check+set}$ formulas as $\langle T', V \rangle$. Since $\langle U, \overline{V} \rangle$ is image finite, the set of processes that $\langle U, \overline{V} \rangle$ can reach by performing a $\lambda$-labelled transition is finite, say $\{\langle U_1, V_1 \rangle, \dots, \langle U_n, V_n \rangle\}$ with $n \in \mathbb{N}$. For every $i \in \{1\dots n\}$, there exist a formula $\phi_i$ and valuation $V_i'$ such that $\langle T', V_i' \rangle \in \ds{\phi_i} \text{ and } \langle U_i, V_i' \rangle \notin \ds{\phi_i}$ or valuations $V_i$ and $V'$ differ for variable $v$. 
	
	We are now in a position to construct a formula that is satisfied by $\langle T, \overline{V} \rangle$ but not by $\langle U, \overline{V} \rangle$, contradicting our assumption that $\langle T, \overline{V} \rangle$ and $\langle U, \overline{V} \rangle$ satisfy the same formulae. 
	$$
	\mbox{We define for each } i \in \{1\dots n\}: refute(i) =
	\left\{
	\begin{array}{ll}
	\setOp{V_i'}\phi_i  & \mbox{if } \langle T', V_i' \rangle \in \ds{\phi_i} \text{ and } \langle U_i, V_i' \rangle \notin \ds{\phi_i} \\
	(v = V'(v)) & \mbox{if the valuations of } V' \mbox{ and } V_i' \mbox{ differ for } v 
	\end{array}
	\right.
	$$
	The formula $\langle \lambda \rangle (refute(1) \land refute(2) \land \dots \land refute(n))$ is satisfied by $\langle T, \overline{V} \rangle$ but not by $\langle U, \overline{V} \rangle$. 
	
\end{proof}

\begin{theorem}\label{thm:relation-logic-state-based}
	Let $\langle P, V \rangle$ and $\langle Q, V \rangle$ be states in some LTS $(S,\TL,\xrightarrow{},s)$ and let all states reachable from $\langle P, V \rangle$ and $\langle Q, V \rangle$ be image-finite. Then $\langle P, V \rangle \bisim{sb} \langle Q, V \rangle$ if and only if for all HML$^{check}$ formulas $\phi$ we have that $\langle P, V \rangle \in \ds{\phi} \Leftrightarrow \langle Q,V \rangle \in \ds{\phi}$.
\end{theorem}
\begin{proof}
	The proof is very similar to the one given for Theorem \ref{thm:relation-logic-stateless}. We will only provide the distinguishing formula. 
	
	Let $\langle T, \overline{V} \rangle \xrightarrow{\lambda} \langle T', V' \rangle$ and let $\{\langle U_1, V_1 \rangle, \dots, \langle U_n, V_n \rangle\}$ be the set of states $\langle U,\overline{V} \rangle$ can reach with a $\lambda$-labelled transition. For every $i \in \{1\dots n\}$, there exists a formula $\phi_i$ such that $\langle T', V' \rangle \in \ds{\phi_i} \text{ and } \langle U_i, V_i \rangle \notin \ds{\phi_i}$ or valuations $V_i$ and $V'$ differ for variable $v$. 
	$$
	\mbox{We define for each } i \in \{1\dots n\}: refute(i) =
	\left\{
	\begin{array}{ll}
	\phi_i  & \mbox{if } \langle T', V' \rangle \in \ds{\phi_i} \text{ and } \langle U_i, V_i \rangle \notin \ds{\phi_i} \\
	(v = V'(v)) & \mbox{if the valuations of } V' \mbox{ and } V_i \mbox{ differ for } v 
	\end{array}
	\right.
	$$
	The formula $\langle \lambda \rangle (refute(1) \land refute(2) \land \dots \land refute(n))$ is satisfied by $\langle T, \overline{V} \rangle$ but not by $\langle U, \overline{V} \rangle$.
\end{proof}

Note that for process expressions we need the set operator in the logic, whereas for states we can only have the check operator on top of regular HML. Intuitively, we need the set operator on the level of process expressions to say something about the behaviour of the process for any valuation.

\section{Translation to mCRL2}\label{sec:translation}
For process algebras without global variables it might be the case that global variables can be modelled using different language constructs. Modelling global variables then often requires a protocol to regulate the access to global variables. In this section we explore how a process expression in our process algebra with global variables can be translated to mCRL2 without introducing extra internal activity. The resulting mCRL2 specification induces an LTS that is isomorphic to the LTS of the original process expression, save some selfloops signalling information on the valuation of that state. We also give a translation from HML$^{check}$ to the modal $\mu$-calculus, which is the logic that is used in mCRL2 to express properties. We show that a HML$^{check}$ formula holds for the original process expression if and only if the translated formula holds for the translated process expression.

\newcommand{\labels}{\ensuremath{\Lambda}}
\newcommand{\names}{\ensuremath{\underline{\Lambda}}}
\newcommand{\multiact}{\ensuremath{\mathbb{M}}}
\newcommand{\multiactnames}{\ensuremath{\underline{\mathbb{M}}}}
\newcommand{\semmultiact}{\ensuremath{\mathcal{M}}}
\subsection{Introduction of mCRL2}
We will introduce the syntax and semantics of the fragment of mCRL2 that is needed to encode global variables. In particular, we will introduce actions parametrised with data and multi-actions. We will not go into the details of the data language itself. It suffices that there exists a semantic interpretation function $\ds{.}$ that maps data expressions to elements of the data domain. We declare a set of data expressions $\mathcal{D}$ and a set of Boolean expressions $\mathcal{B}$ of which the interpretation is an element of $D$ or $\{true,false\}$, respectively. We also presuppose an equality relation $\approx$ on data expressions. For more information on the syntax and semantics of mCRL2 we refer the reader to \cite{sv_book}. Multi-actions in combination with the allow operator  were first proposed in \cite{DBLP:journals/entcs/Weerdenburg08}.

We presuppose a set of \emph{action names} $\underline{\Lambda}$, each with an associated arity. An
\emph{action label} $a(d_1,\dots,d_n)$ consists of an action name $a\in\underline{\Lambda}$ of arity
$n$ and a list of data parameters $d_1,\dots, d_n$. We denote by $\Lambda$ the set of action labels. If
$\alpha\in\Lambda$, then we denote by $\underline{\alpha}$ its name (e.g., $\underline{a(2,3,true)}=a$).

The set of multi-actions \multiact{} is generated by the following grammar:
$$\alpha := \alpha|\alpha\ |\ \tau\ |\ a(d_1,\dots,d_n),$$
where $a(d_1,\dots,d_n) \in \labels$ and $d_1$ to $d_n$ are data expressions or Boolean expressions. Also for each multi-action $\alpha$ we define $\underline{\alpha}$:
$$\underline{\tau} = \tau$$
$$\underline{a(d_1,...,d_n)} = a$$
$$\underline{\alpha|\beta} = \underline{\alpha}|\underline{\beta}.$$
The set of multi-actions where each action label in the multi-action is in \names{} is \multiactnames.

We define a multi-set over $A$, $(A,m)$, where $m: A\rightarrow\mathbb{N}$ is a function assigning a multiplicity to each element of $A$. We define $a \in (A,m)$ to be true if and only if $m(a) > 0$. As notation we use \multiset{} where the elements are listed together with their multiplicity, e.g. \multiset{a:2, b:3}. Over multi-sets $(A,m)$ and $(A,m')$ we define a binary operator \emph{addition}, denoted by $+$, that results in a multi-set $(A,m'')$, where for all $a$ in $A$ it is the case that $m''(a) = m(a) + m'(a)$. Similarly, we define a binary operator \emph{subtraction}, denoted by $-$, such that it results in a multi-set $(A,m'')$, where for all $a\in A$ we have that $m''(a) = max(m(a) - m'(a),0)$. Furthermore, we define inclusion, denoted by $\subseteq$, to hold if and only if for all $a \in A$ we have that $m(a) \leq m'(a)$. For multi-sets over labels we define $\underline{(\labels,m)} = (\names,m')$, where for all $a \in \labels$ it holds that $m(a) = m'(\underline{a})$. 

Given a multi-action $\alpha$ we inductively associate a \emph{semantic multi-action} \ds{\alpha} with it:
$$\ds{\tau} = \multiset{}$$
$$\ds{a(d_1,...,d_n)} = \multiset{a(\ds{d_1},...,\ds{d_n}):1}$$
$$\ds{\alpha|\beta} = \ds{\alpha} + \ds{\beta}$$

The set of all semantic multi-actions is \semmultiact. The set $\mathcal{P}_{mCRL2}$ of process expressions of the fragment of mCRL2 that we need in the translation is generated by the following grammar:
$$P:=\lambda.P\ |\ \delta\ |\ P+P\ |\ P \| P\ |\ \allow{M}{(P)}\ |\ X(d_1,...,d_n)\ |\ \sum_{d:D} P\ |\ \hide{I}{(P)}\ |\ \comm{C}{(P)},$$

\noindent where $\lambda\in\multiact$, $M$ a set of multi-action names, $M \subseteq \multiactnames$, $X$ is a process name, $I$ is a set of action names, $I \subseteq \names$, and $C$ is set of renamings from a set of multi-action names to an action name, notation $a|...|b \rightarrow c$. 

We introduce a function $\commF{C}(\alpha)$, where $\alpha$ is a semantic multi-action that applies communications in $C$ to $\alpha$, e.g. $\comm{\{a|b\rightarrow c\}}{\multiset{a:2,b:3}} = \multiset{b:1,c:2}$. A communication can only be performed when the parameters of action labels match. For the exact semantics of $\commF{C}(\alpha)$ we refer the reader to \cite{sv_book}.

We define a function $\hideF{I}{((\labels,m))}$, such that it results in a multi-set $(\labels,m')$ such that 
$$\forall_{a \in \labels}m'(a) = \left\{
\begin{array}{ll}
0  & \mbox{if } \underline{a} \in I \\
m(a) &\mbox{if } \underline{a} \notin I
\end{array}
\right.$$

The sum operator facilitates a non-deterministic choice over a data domain. For example, in the case the data domain $D$ is the natural numbers, $\sum_{n:D} a(n).P$ can make an $a(0)$ step to process $P[n := 0]$, an $a(1)$ step to process $P[n := 1]$, etcetera.

\begin{table}[h]\small\center
	\begin{tabular}{c} \hline \\
		$\RPref\
		\inference{}{\pref[\alpha]P\step{\ds{\alpha}}P}$
		\quad $\RPar\ \inference{P\step{\alpha}P' & Q\step{\beta}Q' \quad}
		{P\pcmp Q\step{\alpha + \beta}P'\pcmp Q'}$
		\\ \\
		$\RRec\ \inference{P[d_1:=t_,...,d_n:=t_n]\step{\alpha}P' & X(d_1:D_1,...,d_n:D_n)\defeqn P}{X(t_1,...,t_n)\step{\alpha}P'}$
		\\ \\
		$\RSuml\
		\inference{P\step{\alpha}P'}%
		{{P\acmp Q}\step{\alpha}P'}$\quad
		$\RSumr\
		\inference{Q\step{\alpha}Q'}%
		{{P\acmp Q}\step{\alpha}Q'}$
		\\ \\
		$\RSum\ \inference{P[d:=t_e]\step{\alpha}P' & t_e \in \mathcal{D}}{\sum_{d:D}P\step{\alpha}P'}$	
		\quad $\RHide\ \inference{P \step{\alpha}P'}{\tau_I(P)\step{\theta_I(\alpha)}\tau_I(P')}$
		\\ \\
		$\RParl\ \inference{P\step{\alpha}P'}{P\pcmp
				Q\step{\alpha}P'\pcmp Q}$\quad
		$\RParr\ \inference{Q\step{\alpha}Q'}{P\pcmp Q\step{\alpha}P\pcmp Q'}$
		\\ \\
		$\RComm \inference{P \step{\alpha} P'}{\Gamma_C(P)\step{\gamma_C(\alpha)}\Gamma_C(P')}$ 
		\quad $\RAllow\ \inference{P \step{\alpha}P'}{\nabla_M(P)\step{\alpha}\nabla_M(P')}\ds{\underline{\alpha}} \in M$
		\\ \\
		\hline
	\end{tabular}
	\caption{Structural operational semantics of our fragment of mCRL2.}\label{tab:os-mcrl2}
\end{table}

An LTS $(S,\semmultiact,\xrightarrow{},P)$ can be associated to a process expression $P$. The set of states is the set of process expressions, $S = \mathcal{P}_{mCRL2}$. The set of transitions is generated by the proof system based on the structural operation semantics (see Table \ref{tab:os-mcrl2}). 

\subsection{Translation of process expressions and valuations}
Recall that a specification in the process algebra with global variables consists of the following: a data domain $D$, a set of variable names $\Var{}$, a set of action labels $\Act$, a set of process names $\PN$ and their defining equations, a communication function $\gamma$, and an initial state consisting of a process expression and an initial valuation $V$. We consider a restricted grammar for the translation. 

\paragraph{Sequential components} The set of \emph{sequential process expressions} $\mathcal{P}_{Seq}$
is generated by the following grammar (with $v$ ranging over $\Var{}$, $d$ ranging over $D$, $X$ ranging over
$\PN$ and $\lambda$ ranging over $\TL$):
$$Seq:= Seq+Seq\ |\ (v = d) \xrightarrow{} Seq \ |\ \lambda.Seq\ |\ \delta\ |\ \lambda.X.$$ 

By a \emph{sequential recursive specification} $E$ we mean a set of defining equations
$X\defeqn t,$
with $t$ a sequential process expression, including precisely one such
equation for every $X\in\PN$.

\paragraph{Parallel-sequential processes}

Presupposing a sequential recursive specification
$E$, the set of \emph{parallel-sequential} process expressions $\mathcal{P}_{\mathit{Par}}$ over $E$ is
generated by the following grammar (with $X$ ranging over
$\PN$ and $Seq$ ranging over sequential process expressions):
$$Par:= Par \| Par\ |\ X \ |\ Seq.$$

We assume that the recursive specification $E$ is sequential and that the process expression under consideration for translation is of the shape $\partial_B(P)$, where $P$ is parallel-sequential process expression and $B \subseteq \Act$. For the sake of readability, in our explanations below we restrict our attention to the case that there is one global variable $g$. In Section \ref{sec:discussion} we explain how to generalise the translation to any number of variables. Now that the input for the translation is clear, we show how it is translated to mCRL2.

\noindent The value of global variables is tracked by a dedicated process $Globs$, defined below. \\
$Globs(d: D) =$\\
\indent $    checkG(d,true).Globs(d)$\\
\indent    $+ checkG(d,true)|checkG(d,true).Globs(d)$\\
\indent    $+ \Sigma_{new:D}. checkG(d,true)|assignG(g,new).Globs(new)$\\
\indent   $+ value(g,d).Globs(d);$\\
The process can communicate the current value of the global variable with a \texttt{checkG} action, of which the first parameter is of type $D$ and the second a constant of type $Bool$. It can perform a \texttt{checkG} action twice in a multi-action to facilitate informing two parallel processes in one step. Since our process algebra with global variables only allows handshaking communication there can never be more than two parallel processes that participate in a transition. It facilitates changing the value of the global variable with an \texttt{assignG} action, with one parameter of type $D$ carrying the new value. It can emit the current value of a global variable with a \texttt{value} action, with a single parameter of type $D$. 

\newcommand{\trans}{\ensuremath{\chi}}
\newcommand{\transP}{\ensuremath{\Psi}}
\newcommand{\CS}{\epsilon}
\newcommand{\powerset}{\ensuremath{\wp}}
We translate the recursive specification $E$ to a recursive mCRL2 specification $E'$, which includes defining equations for all the process names in $\PN$ and additionally a defining equation for $Globs$. Let $\powerset(A)$ denote the powerset of $A$. We introduce a function $\trans: \mathcal{P}_{Seq} \times \powerset(D) \rightarrow \mathcal{P}_{mCRL2}$, which we will define shortly. For every defining equation $X \defeqn t$ in $E$ there is a defining equation $X \defeqn \trans(t,\emptyset)$ in $E'$. The function \trans{} is defined below, where $\CS \subseteq D$ is a set of constraints on the global variable that is eventually transformed into an appropriate $checkP$ action.

\begin{center}
	\begin{tabular}{llp{11cm}}
$\trans(P_1 + P_2,\CS) $ & $=$ & $ \trans(P_1,\CS) + \trans(P_2,\CS)$\\
$\trans((g = d) \xrightarrow{} P_1,\CS) $ & $=$ & $ \trans(P_1,\CS \cup \{d\})$\\
$\trans(a.P_1,\CS) $ & $=$ & $ (\sum_{d_1:D} a|checkP(d_1,\bigwedge_{d \in \epsilon} d_1 \approx d)).\trans(P_1,\emptyset)$\\
$\trans(assign(g, d').P_1,\CS) $ & $=$ & $ (\sum_{d_1:D} assignP(g,d')|checkP(d_1,\bigwedge_{d \in \epsilon} d_1 \approx d)).\trans(P_1,\emptyset)$\\
$\trans(X,\CS)$ & $=$ & $X$\\
$\trans(\delta,\CS)$ & $=$ & $\delta$\\
	\end{tabular}
\end{center}

We define a set of communications $C_{\gamma}$, such that $a|b->c \in C_{\gamma}$ or $b|a->c \in C_{\gamma}$ if and only if $\gamma(a,b) = c$ (we should include only one of $a|b \rightarrow c$ and $b|a \rightarrow c$ in $C_\gamma$ to
satisfy the requirement that the left-hand sides of communications in $C_\gamma$ are disjoint). We define an extended set of communications that includes communications with $Globs$: $C_{G\gamma} =  C_{\gamma} \cup \{assignP|assignG->assign,checkP|checkG->check\}$. Given a set of encapsulated actions $B$ we define a set of allowed actions $A_B = (\Act\setminus B) \cup \{value,assign\}$. We extend \trans{} to parallel-sequential process expressions in the following way. 

\begin{center}
	\begin{tabular}{llp{11cm}}
$\trans(P_1||P_2,\emptyset) $ & $=$ & $ \trans(P_1,\emptyset)||\trans(P_2,\emptyset)$
	\end{tabular}
\end{center}

We translate the process expression $\partial_B(P)$, with an initial valuation $V, V(g) = d$, to the mCRL2 process expression $\allow{A_B}(\hide{\{check\}}(\comm{C_{G\gamma}}(\trans(P,\emptyset)||Globs(d))))$, which we abbreviate to $\transP(P,V)$. 

\subsection{Translation of formula}
\newcommand{\fLogic}[1]{\ensuremath{\theta(#1)}}
\newcommand{\fLogicInv}[1]{\ensuremath{\theta^{-1}(#1)}}
The selfloops labelled with $value$ provide information on the values of global variables in every state, which we will exploit in the translation of HML$^{check}$ formulas. Given a HML$^{check}$ formula we eliminate each occurrence of the check operator of the shape $(v = e)$ by substituting it with $\langle value(v,e) \rangle true$. We denote this substitution function with $\theta$, which we define inductively:\\
\begin{center}
	\begin{tabular}{lll}
		$\fLogic{true} $&=&$true$, \\
		$\fLogic{false} $&=&$false$, \\
		$\fLogic{v=e} $&=&$ \langle value(v,e) \rangle true, $\\
		$\fLogic{\neg\phi} $&=&$ \neg\fLogic{\phi},$\\
		$\fLogic{\phi_1 \land \phi_2} $&=&$ \fLogic{\phi_1} \land \fLogic{\phi_2}, $\\
		$\fLogic{\phi_1 \vee \phi_2} $&=&$\fLogic{\phi_1} \vee \fLogic{\phi_2}, $\\
		$\fLogic{\langle T \rangle \phi} $&=&$ \langle T \rangle \fLogic{\phi}, $\\
		$\fLogic{[T] \phi} $&=&$ [T]\fLogic{\phi}$.
	\end{tabular}
\end{center}

\subsection{Correctness of translation}
From here on, when we consider the translation of some state $\state{P}{V}$ to $\transP(P,V)$ we assume that the context of the process expression, such as the data domain $D$, the set of actions and a recursive specification have been encoded in mCRL2 as described in the previous section. 

We will prove that a $HML^{check}$ formula $\phi$ holds in a state $\state{P}{V}$ if and only if $\fLogic{\phi}$ holds for $\transP(P,V)$. To achieve this we use a stepping stone. In Definition \ref{def:variable-consistent} we define a relation between LTSs with and without a valuation function in the state, called \emph{variable consistency}. We prove that the LTSs induced by $\state{P}{V}$ and $\transP(P,V)$ are variable consistent, which we use in Theorem \ref{thm:satisfy-same-formula} to prove that any $HML^{check}$ formula $\phi$ holds for $\state{P}{V}$ if and only if $\fLogic{\phi}$ holds for $\transP(P,V)$
\newcommand{\linkState}[1]{\ensuremath{\ell(#1)}}
\newcommand{\linkStateInv}[1]{\ensuremath{Link^{-1}(#1)}}

\begin{definition}\label{def:variable-consistent}
	Let $\mathcal{L}_1=(S_1,TL_1,\rightarrow_1,s_1)$ be an LTS such that $S_1= \mathcal{P}
	\times\mathcal{V}$, and let $\mathcal{L}_2=(S_2,TL_2,\rightarrow_2,s_2)$ be an LTS such that $S_2 = \mathcal{P}_{mCRL2}$.
	We say that $\mathcal{L}_2$ is \emph{variable-consistent} with $\mathcal{L}_1$ if there exists a
	mapping $\ell: S_1\rightarrow S_2$ such that whenever some state $s_1’$ is reachable from $s_1$ in $\mathcal{L}
	_1$, then $\ell(s_1’)$ is reachable from $\ell(s_1)$ in $\mathcal{L}_2$ and
	\begin{enumerate}
		\item for all states $s_1',s_2' \in S_2$ reachable from $s_2$ and such that $s_1' \step{\lambda} s_2'$ we have that $ \lambda \in TL \cup \{value(v,d) \mid v\in \Var{} \land d\in D\}$,
		\item for all $\langle P,V \rangle \in S_1, s' \in S_2, v \in \Var{}, d \in D$ we have that $\linkState{\langle P,V \rangle} \xrightarrow{value(v,d)} s'$ if and only if $V(v) = d$ and $\linkState{\langle P,V \rangle} = s'$,
		\item for all $\lambda \in \TL{}_1$ and reachable states $s_1', s_2' \in S_1$ we have that $s_1' \xrightarrow{\lambda} s_2'$ if and only if $\linkState{s_1'} \xrightarrow{\lambda} \linkState{s_2'}$.
	\end{enumerate}
\end{definition}

For the first property of variable consistency we prove the following lemma.
\begin{lemma}\label{lem:consistency-1}
	For all parallel-sequential process expressions $P$, process expression $P' \in \mathcal{P}_{mCRL2}$, valuations $V$, $\alpha \in \multiact$ and $B\subseteq \Act$ we have that $\transP(P,V) \step{\alpha} P'$ implies $\alpha \in TL \cup \{value(g,d) \mid d\in D\}$.
\end{lemma}
\begin{proof}
	This follows immediately from the allow operator in $\transP(P,V)$, which does not allow multi-actions that are not in $TL \cup \{value(g,d) \mid d\in D\}$.
\end{proof}

\noindent Towards proving the second property of variable consistency we prove the following lemma.
\begin{lemma}\label{lem:consistency-2}
	For all parallel sequential process expressions $P$, process expression $P' \in \mathcal{P}_{mCRL2}$, valuation $V$, $d\in D$, $B \subseteq \Act$ we have that $\transP(P,V) \xrightarrow{value(g,d)} P'$ if and only if $V(g) = d$ and $\transP(P,V) = P'$.
\end{lemma}
\begin{proof}
	The process expression $\transP(P,V)$ contains a parallel component $Globs(d)$. The $Globs$ process can make a $value(g,d)$ transition where $V(g) =d$. Moreover, all such $value(g,d)$ transitions are self-loops. Finally, the $Globs$ process is the only sub process in $\transP(P,V)$ that is able to produce a $value$ transition.
\end{proof}

\noindent Towards proving the third property of variable consistency we first provide a number of auxiliary lemmas.

\begin{lemma}\label{lem:seq-end-in-translation}
	For all sequential process expressions $P$, process expression $P' \in \mathcal{P}_{mCRL2}$, $\lambda \in \Act \cup \{assignP(g,d) \mid d \in D\}$, $d_1 \in D$ we have that $\trans(P,\emptyset) \xrightarrow{\lambda|checkP(d_1,true)} P'$ implies that there exists $P''$ such that $\trans(P'',\emptyset) = P'$.
\end{lemma}
\begin{proof}
	From the definition of $\trans$ it follows that if $\trans(P,\emptyset)$ can make a $\lambda|checkP(d_1,true)$ labelled transition then there exists some $Q$ and $P_1$ such that $\trans(P,\emptyset) = Q + (\sum_{d_1:D} \lambda|checkP(d_1,\bigwedge_{d \in \epsilon} d_1 \approx d)).\trans(P_1,\emptyset)$. Hence after making the $\lambda|checkP(d_1,true)$ labelled transition we end up in $\trans(P_1,\emptyset)$.
\end{proof}

\begin{lemma}\label{lem:parseq-end-in-translation}
	For all parallel-sequential process expressions $P$, $\alpha \in \multiact$ we have that $\transP(P,V) \xrightarrow{\alpha} P'$ implies that there exists $P''$ and $V'$ such that $\transP(P'',V') = P'$.
\end{lemma}
\begin{proof}
	By Lemma \ref{lem:consistency-1} we conclude that $\alpha \in TL \cup \{value(g,d) \mid d\in D\}$. In the case that $\alpha$ is a $value$ transition it is a selfloop and ends in $\transP(P,V)$. In any other case $\transP(P,V)$ makes a step that includes a contribution from one or more of the parallel components of $P$. From the definition of $\trans$ it follows that any contribution of a parallel component is of the shape $\lambda|checkP(d_1,b)$, where $d_1 \in D$ and $b \in \{true,false\}$. The $checkP$ must communicate with a $checkG$, otherwise the action will be blocked by the allow operator. Hence $b=true$, enabling us to use Lemma \ref{lem:seq-end-in-translation} to conclude that for every parallel component contributing to $\alpha$ there exists some process expression $P_{seq}$ such that the parallel component ends in $\trans(P_{seq},\emptyset)$. The parallel components of $\transP(P,V)$ that do not contribute to $\alpha$ remain in a shape such that there exists some process expression $P_{seq}$ such that the parallel component is $\trans(P_{seq},\emptyset)$. The $Globs$ process remains unchanged or its valuation is updated, in which case there exists some valuation $V'$ that reflects the updated value. The allow, hide and communication operators remain unchanged. Hence, after any $\alpha$ step $\transP(P,V)$ ends in a state $\transP(P'',V')$.
\end{proof}

\begin{lemma}\label{lem:relation-sequential}
	For all sequential process expressions $P,P'$, $DD \subseteq \mathcal{D}$, $a \in \Act$ and $assign(g,d')\in\TL$ we have that $\forall_{V \in \mathcal{V}} (\bigwedge_{d \in DD} V(g) = \ds{d}) \implies \state{P}{V} \step{a} \state{P'}{V}$ if and only if $\exists_{d_1 \in D}\trans(P,\emptyset) \xrightarrow{a|checkP(d_1,true)} \trans(P',\emptyset) \land \bigwedge_{d \in DD} d_1 = \ds{d}$ and we have that $\forall_{V \in \mathcal{V}} (\bigwedge_{d \in DD} V(g) = \ds{d}) \implies \state{P}{V} \xrightarrow{assign(g,d')} \state{P'}{V[g \mapsto d']}$ if and only if $\exists_{d_1 \in D}\trans(P,\emptyset) \xrightarrow{assignP(g,d')|checkP(d_1,true)} \trans(P',\emptyset) \land \bigwedge_{d \in DD} d_1 = \ds{d}$.
\end{lemma}
\begin{proof}
	This can be proven by induction on the structure of $P$, the induction hypothesis is that the bi-implication holds for every direct subprocess of $P$ and for every defining equation of process names. The key insight is the second field of the $checkP$ action is only true when the condition for the data value in the first field of $checkP$, constructed by $\trans$, is true.
\end{proof}

\begin{lemma}\label{lem:step-par-seq}
	For any parallel-sequential process expression $P$, process expression $P'$, $\lambda\in\TL$ and valuations $V,V'$ we have that $\state{\partial_B(P)}{V} \step{\lambda} \state{\partial_B(P')}{V'}$ implies $P'$ is again a parallel-sequential process expression.
\end{lemma}
\begin{proof}
	Any step made from $\partial_B(P)$ leaves the $\partial_B$ operator and the parallel composition intact. One or more of the parallel components make a step. By the structure of parallel-sequential process expressions these parallel components are either a process name or a sequential process. Since we also assume that the defining equations of every process name is a sequential process expression the process name will make a step as such. By the structure of sequential process expressions they can make a step to a sequential process expression or a process name. Hence, after any step $P$ is again a parallel composition with process names and sequential process expressions.
\end{proof}

\begin{lemma}\label{lem:consistency-3}
	For all valuations $V$ and $V'$, $\lambda \in \TL$, parallel-sequential process expressions $P$, $B \subseteq \Act$ we have that $\langle \partial_B(P),V \rangle \xrightarrow{\lambda} \langle \partial_B(P'),V' \rangle$ if and only if $\transP(P,V) \xrightarrow{\lambda} \transP(P',V')$. 
\end{lemma}
\begin{proof}
	Both directions of the bi-implication can be proven with a case distinction on the type of transition using three cases: an action from $\Act$ stemming from one of the parallel components, a handshake stemming from two parallel components and an assignment. To prove the implication from left to right Lemma \ref{lem:relation-sequential} can be used to prove that for each contribution by a parallel component of $\langle \partial_B(P),V \rangle$ the step can be matched with an appropriate step from a parallel component of $\transP(P,V)$. Lemma \ref{lem:step-par-seq} ensures that after taking a transition we end in the translation of a parallel-sequential process again. To prove the implication from right to left Lemma \ref{lem:seq-end-in-translation} and Lemma \ref{lem:relation-sequential} can be used to prove that for each contribution by a parallel component of $\transP(P,V)$ the step can be matched with an appropriate step from a parallel component of $\langle \partial_B(P),V \rangle$.
\end{proof}

\begin{theorem}\label{thm:translation-variable-consistent}
	For all parallel-sequential process expressions $P$, valuations $V$, we have that the LTSs induced by \state{P}{V} and $\transP(P,V)$ are variable consistent.
\end{theorem}
\begin{proof}
	For every state $\state{P'}{V'}$ reachable from \state{P}{V} we define $\linkState{\state{P'}{V'}} = \transP(P',V')$. Lemma \ref{lem:consistency-1} proves condition 1, Lemma \ref{lem:consistency-2} proves condition 2 and Lemma \ref{lem:consistency-3} together with Lemma \ref{lem:parseq-end-in-translation} proves condition 3.  
\end{proof}

\begin{corollary}\label{cor:relation-sbb-sb}
	For all parallel-sequential process expressions $P,Q$ and valuations $V_1,V_2$ we have that $\state{P}{V_1} \bisim{sb} \state{Q}{V_2}$ if and only if $\transP(P,V_1) \bisim{} \transP(Q,V_2)$.
\end{corollary}
\begin{proof}
	This follows immediately from the definition of variable consistency. The difference between state-based bisimilarity and strong bisimilarity is only that state-based bisimilarity requires that the valuation in states is equal. By condition 2 of variable consistency the valuations $V_1$ and $V_2$ are equal if and only if $\transP(P,V_1)$ and $\transP(Q,V_2)$ have the same $value$ labelled self-loops on states.
\end{proof}

\begin{theorem}\label{thm:satisfy-same-formula}
	Let $(S,\TL{},\xrightarrow{},s)$ be an LTS where $S = \mathcal{P}\times \mathcal{V}$, let $(S',\TL{}',\xrightarrow{}',s')$ be an LTS where $S' = \mathcal{P}_{mCRL2}$ and let these two LTSs be variable consistent. A HML$^{check}$ formula $\phi$ holds in some state $\langle P,V \rangle \in S$ if and only if \fLogic{\phi} holds in $\linkState{\langle P,V \rangle} \in S'$.
\end{theorem}
\begin{proof}
The proof is by induction on the structure of $\phi$ with the induction hypothesis that any subformula $\phi'$ of $\phi$ holds for $\langle P,V \rangle \in S$ if and only if \fLogic{\phi'} holds for $\linkState{\langle P,V \rangle} $. In the case $\phi = (v=e)$ condition 2 of variable consistency is necessary to relate the valuation in a state and the $value$ labelled selfloops. Condition 3 of variable consistency is needed for the case $\phi = \langle T \rangle \phi'$ and $\phi = [ T ] \phi'$ to show that transitions can be mimicked. Furthermore, in the case $\phi = [ T ] \phi'$ we also need condition 1 of variable consistency to show that \linkState{\langle P,V \rangle} does not have more $\lambda$ labelled transitions.
\end{proof}

\section{Discussion}\label{sec:discussion}
For the encoding in mCRL2 and subsequent correctness proofs we have made the assumption that there is only one global variable, which is rather restrictive. To generalize the translation to handle any number of global variables we would need to adjust the following. The $Globs$ process should be adjusted to track more global variables by making the parameter of the process a mapping from variable names to values. Upon performing an $assignG(v,d)$ action $Globs$ should update the mapping such that $v$ maps to $d$. To communicate the values of global variables in a $check(d_1,\dots,d_n,true)$ action we need an ordering on the global variables: $d_1$ is given the value of variable one, $d_2$ is given the value of variable two, etcetera. The condition determining the last parameter of the $checkP$ action should also be adjusted to use this ordering, e.g. when $\trans$ gathers a requirement $(v,d)$ and variable $v$ is the $i$th variable then the condition in $checkP$ should include a conjunct $d_i \approx d$.

We intend to continue researching process algebras with global variables. One research direction is to extend mCRL2 with global variables. The simple process algebra presented in this paper only allows for very simple conditions on global variables: checking whether a variable has a specific value. If global variables could be integrated into mCRL2 we could use its powerful data language to specify complex conditions. We would also like to research scoped shared variables, including creation and scope extrusion.

\section{Conclusion}\label{sec:conclusion}
In this paper we have presented a simple process calculus with global variables and studied various aspects of it. To start we examined appropriate notions of equivalence: stateless bisimulation for process expressions and state-based bisimulation for states. Then, for our first contribution we presented a logic extending HML with a check and a set operator and proved that HML$^{check}$ is strong enough to differentiate states that are not state-based bisimilar and HML$^{check+set}$ is strong enough to differentiate process expressions that are not stateless bisimilar. Finally, for our second contribution we give a translation to mCRL2, using the multi-action concept, preserving HML$^{check}$ formulas. Translating to mCRL2 allows us to reuse the already existing tools. The translation mostly preserves the syntactic structure and, in particular, the parallel composition (adding one extra parallel process).

When analysing whether a distributed system satisfies a liveness property, it is necessary to define through a completeness criterion which runs of the system should be considered in the analysis. Recently, \emph{justness} was proposed as a suitable completeness criterion that takes into account the component structure of the system  \cite{GH15} and excludes unrealistic runs. Modelling shared variables as separate components hampers a straightforward definition of justness \cite{DGH17,BLW20}. Since global variables need not be modelled as separate components in the process algebra proposed in Section~\ref{sec:GlobalVarDef}, it may facilitate a more elegant analysis of liveness properties under justness assumptions for distributed systems that rely on shared variables for the communication between components.

\subsection*{Acknowledgements}
For the presentation of the semantics of mCRL2, and in particular for the semantics of multi-actions, we have benefited from work by Maurice Laveaux. We would also like to extend our gratitude to the anonymous reviewers. Their comments led to Corollary \ref{cor:relation-sbb-sb}.
    
\bibliographystyle{eptcs}
\bibliography{bibliography.bib}

\end{document}